\documentclass[a4paper,11pt]{article}

\usepackage{jheppub} 

\usepackage[T1]{fontenc} 

\usepackage{graphicx} 
\usepackage{amsthm}

\newtheorem{thm}{Theorem}[section]

\newtheorem{defn}[thm]{Definition}

\newcommand{\abs}[1]{\left| #1 \right|} 
\newcommand{\ket}[1]{| #1 \rangle } 
\newcommand{\bra}[1]{\langle  #1 |} 
\newcommand{\braket}[2]{\langle  #1 \vphantom{#2} | #2 \vphantom{#1} \rangle } 
\newcommand{\diracprod}[3]{\left\langle  #1 \vphantom{#2#3} \right|#2 \left| #3 \vphantom{#1#2} \right\rangle } 
\let\baraccent=\= 
\renewcommand{\=}[1]{\stackrel{#1}{=}} 

\newcommand{\im}{\text{Im}}

\def\A{\mathcal{A}}
\def\B{\mathcal{B}}
\def\N{\mathcal{N}}

\def\O{\mathcal{O}}
\def\P{\mathcal{P}}

\def\L{\mathcal{L}}
\def\H{\mathcal{H}}

\def\C{\mathbb{C}}

\def\Tr{\text{Tr}}

\def\Srel{S_{\text{rel}}}

\title{Adding the algebraic Ryu-Takayanagi formula to the algebraic reconstruction theorem}

\author[a]{Mingshuai Xu}
\author[a]{Haocheng Zhong}

\affiliation[a]{Shing-Tung Yau Center and School of Physics, Southeast University, Nanjing 210096, China}

\emailAdd{xumingshuai@seu.edu.cn, zhonghaocheng327@gmail.com}

\abstract{A huge progress in studying holographic theories is that holography can be interpreted via the quantum error correction, which makes equal the entanglement wedge reconstruction, the Jafferis-Lewkowycz-Maldacena-Suh formula, the radial commutativity and the Ryu-Takayanagi formula. We call the equivalence the reconstruction theorem, whose infinite-dimensional generalization via algebraic language was believed to exclude the algebraic version of the Ryu-Takayanagi formula. However, recent developments regarding gravitational algebras have shown that the inclusion of the algebraic Ryu-Takayanagi formula is plausible. In this letter, we prove that such inclusion holds for the cases of type I/II factors, which are expected to describe holographic theories.}

\begin{document}
\maketitle
\flushbottom

\section{Introduction}

One of the milestones in understanding quantum gravity in the last few decades is the AdS/CFT correspondence \cite{Maldacena:1997re,Witten:1998qj,Gubser:1998bc}, which conjectures an equivalence between $d-$dimensional quantum gravity on anti-de Sitter (AdS$_d$) spacetime and $(d-1)-$dimensional conformal field theory (CFT$_{d-1}$). Practically, the dual CFT$_{d-1}$ holographically lives at the boundary of the AdS$_d$ which is hence denoted as the bulk theory. The equivalence between the boundary theory and the bulk theory is realized as several statements about correspondences between quantities in the respective theories, which can be summarized as follows,
\begin{itemize}
	\item Entanglement wedge reconstruction (or subregion duality) \cite{Czech:2012bh,Hamilton:2006az,Morrison:2014jha,Bousso:2012sj,Bousso:2012mh,Hubeny:2012wa,Wall:2012uf,Headrick:2014cta,Jafferis:2015del,Dong:2016eik}: given a subregion of the boundary, one is able to reconstruct the entanglement wedge of the boundary subregion. To be specific, any bulk operator inside the entanglement wedge can be reconstructed via the information on the boundary subregion.
	
	\item  Jafferis-Lewkowycz-Maldacena-Suh (JLMS) formula \cite{Jafferis:2015del}: given two density operators (or density matrices) $\rho_{A},\sigma_A$ restricted on the boundary subregion $A$ whose entanglement wedge in the bulk is denoted by $a$, we have 
	\begin{equation}
		\Srel(\rho_A|\sigma_A)=\Srel(\tilde{\rho}_a|\tilde{\sigma}_a)
	\end{equation}
	where $\Srel$ is the quantum relative entropy and $\tilde{\rho}_a,\tilde{\sigma}_a$ are density operators in $a$ which dual to $\rho_{A},\sigma_A$ respectively.

	\item Radial commutativity \cite{Polchinski:1999yd,
		Almheiri_2015,Harlow:2018fse}: any bulk operator at a bulk time-slice should commute with all boundary operators localized at the boundary of that time-slice.
	
	\item Ryu-Takayanagi (RT) formula \cite{Ryu:2006bv,Hubeny:2007xt,Casini_2011,
		Lewkowycz:2013nqa,Nishioka:2018khk,Faulkner:2013ana,Engelhardt:2014gca}\footnote{Some refer to the Ryu-Takayanagi (RT) or Hubeny-Rangamani-Takayanagi (HRT) formula as the one without bulk correction, and call \eqref{eq:RT} the quantum extremal surface (QES) formula. Accordingly, the RT/HRT surface is also called the quantum extremal surface.}: given a density operator $\rho_{A}$ on the boundary subregion $A$ and the dual density operator $\tilde{\rho}_a$ in the entanglement wedge $a$, their von Neumann entropies satisfy
	\begin{equation}\label{eq:RT}
		S(\rho_A)=\mathcal{L}_A+S(\tilde{\rho}_a)
	\end{equation}
	where $\mathcal{L}_A$ is the area of the Ryu-Takayanagi surface of $A$ over $4G$.
\end{itemize}
Historically, these topics have been developed rather independently. However, one major progress in understanding holographic theories is that the above statements can be put on an equal footing by using the language of quantum error correction (QEC) \cite{Grassl:1996eh}. Furthermore, the equivalence between the above statements is independent of the specific details of the holographic theories. In this paper, we call the equivalence the \emph{reconstruction theorem}. Early proof of the reconstruction theorem assumes the holographic models are finite-dimensional \cite{Almheiri_2015,Dong:2016eik,Harlow:2016vwg}\footnote{See also appendix A of \cite{Zhong:2024fmn} for a summary proof, and \cite{Faulkner:2020hzi,Akers:2021fut,Gesteau:2023hbq} for an incomplete list of related developments.}, which is later generalized into infinite-dimensional cases \cite{Kamal:2019skn,Kang:2018xqy}.

In studying infinite-dimensional quantum theories, von Neumann algebra serves as a natural mathematical tools. Especially for quantum field theory, von Neumann algebra provides an axiomatic way to formalize quantum field theory, whose mathematical framework is called algebraic quantum field theory (AQFT).  In
AQFT, one focus on observables rather than quantum states, and observables of a quantum system form a von Neumann algebra after assuming some physical conditions like causality, Poincar\'e symmetry, etc. We call the algebraic generalization of the reconstruction theorem by using the language of von Neumann algebra the \emph{algebraic reconstruction theorem}.

Unlike the quantum relative entropy, the von Neumann entropy is ill-defined in infinite-dimensional cases, which further implies that the algebraic generalization of the RT formula is problematic. As we will briefly discuss, von Neumann algebras are classified into factors of three types labeled by type I/II/III, and the von Neumann entropy is only ill-defined for type III factors which naturally occur in quantum field theory \cite{Yngvason:2004uh}, so it was believed that the algebraic reconstruction theorem should exclude the algebraic version of the RT formula \cite{Kang:2018xqy,Kang:2019dfi}. However, recent developments in studying algebraic implications in quantum gravity have shown that it is possible to describe gravitational algebras without type III factors \cite{Witten:2021unn,Chandrasekaran:2022cip,Jensen:2023yxy,Kudler-Flam:2023hkl,Colafranceschi:2023urj,Faulkner:2024gst,Kudler-Flam:2023qfl,Kudler-Flam:2024psh,AliAhmad:2023etg}, in which cases the algebraic RT formula is plausible. Motivated by this, we refine the algebraic reconstruction theorem by including the algebraic RT formula, whose proof is the main work of the paper.

The paper is organized as follows. In section \ref{sec:pre}, we introduce the necessary basics of von Neumann algebra and the modular theory (or Tomita-Takesaki theory). The modular theory is an important tool in describing entanglement via the algebraic language, which is discussed in section \ref{sec:alg entropies}. In section \ref{sec:rec}, we first review the ordinary reconstruction theorem in finite-dimensional cases, then present the refined algebraic reconstruction theorem with the proof. We end with section \ref{sec:discussions} by giving further discussions about the algebraic RT formula.


\section{Preliminaries}\label{sec:pre}

In this section, we give a very brief review and establish our notations for von Neumann algebra and modular theory (or Tomita-Takesaki theory), which are necessary for the algebraic reconstruction theorem. This section basically follows \cite{Witten:2018zxz,Kang:2018xqy,Kang:2019dfi}, and readers who are interested in more rigorous details are encouraged to consult, for examples \cite{Vaughan2009,Sorce:2023fdx,reed1972methods,takesaki2006tomita}.

\subsection{Basics of von Neumann algebra}\label{subsec:vN alg}

	%
	%
	%
	%
	%
	%
	%
	%
	%

%
	%
	%
	%

\begin{defn}
	A linear operator on a Hilbert space $\H$ is a linear map from (a subspace of) $\H$ into $\H$. The set of all such operators is denoted by $\L(\H)$.
\end{defn}

\begin{defn}
	A bounded operator is a linear operator $\O$ satisfying $||\O\ket{\psi}||\leq K||\ket{\psi}||,~\forall\ket{\psi}\in\H$ for some $K\in\mathbb{R}$. The infimum of all such $K$ is called the norm of $\O$. The algebra of all bounded operators on $\H$ is denoted by $\B(\H)\subset \L(\H)$.
\end{defn}

\begin{defn}
	The commutant of a subset $S\subset \B(\H)$ is a subset $S'\subset \B(\H)$ defined by
	\begin{equation}
		S'\equiv\{\O\in\B(\H)|[\O,\P]=0,~\forall \P\in S\}
	\end{equation}
	i.e. every element in $S'$ commutes with all elements in $S$.
\end{defn}

	%
	%
	%

\begin{defn}
	The hermitian conjugate (or adjoint) of an operator $\O$ is an operator $\O^\dagger$ satisfying $\braket{\psi}{\O\xi}=\braket{\O^\dagger \psi}{\xi}$. A hermitian (or self-adjoint) operator $\O$ satisfies $\O=\O^\dagger$.
\end{defn}

\begin{defn}
	A von Neumann algebra on $\H$ is a subalgebra $\A\subset\B(\H)$ satisfying
	\begin{itemize}
		\item $\mathbb{I}\in \A$,
		
		\item $\A$ is closed under hermitian conjugation,
		
		\item $\A''=\A$.
		
	\end{itemize} 
\end{defn}

One should notice that the commutant of a von Neumann algebra is also a von Neumann algebra itself: $(\A')''=(\A'')'=\A'$. Another remark we would like to emphasize here is that for quantum systems, the algebra of observables (hermitian operators) forms a von Neumann algebra. Combining the two facts implies that von Neumann algebras are associated with subregions of physical systems, as we will demonstrate by examples in later sections. Furthermore, two casually independent subregions possess the algebras of localized observables being commutants of each other, as required by causality.

\begin{thm}
	Any operator in a von Neumann algebra is a linear combination of two hermitian operators or four unitary operators.
\end{thm}

\begin{proof}
	See for examples \cite{Kang:2018xqy} or page 19 in \cite{Vaughan2009}.
\end{proof}

\begin{defn}
	A von Neumann algebra $\A$ is a factor if it has a trivial center $\mathcal{Z}$:
	\begin{equation}
		\mathcal{Z}\equiv \A\cap\A'=\{\lambda \mathbb{I}|\lambda\in\C\}
	\end{equation}
	otherwise $\A$ is called a non-factor.
\end{defn}

In fact, any non-factor can be ``decomposed'' into factors \cite{Sorce:2023fdx}. Especially in finite-dimensional cases \cite{Harlow:2016vwg}, non-factors can always be decomposed into a ``block-diagonal'' form by choosing an appropriate basis, with each block being a factor. In other words, when consider classifications of von Neumann algebra, we only need to consider factors, which are classified into three types: type I/II/III. In this paper, we do not need the explicit classifications of factors. Instead, we are interested in whether some notions are well-defined or not in different types of factors, as summarized as follows,

$$
\begin{array}{|c|c|c|c|c|c|}
	\hline {\text {Type}} & \H=\H_A \otimes \H_B & \Tr & \rho_\psi & S(\psi;\A)  & S_{\text {rel }}(\psi|\xi;\A) \\
	\hline \text { I } & \checkmark & \checkmark & \checkmark & \checkmark  & \checkmark \\
	\hline \text { II } & \times & \checkmark & \checkmark & \checkmark &  \checkmark \\
	\hline \text { III } & \times & \times & \times & \times  & \checkmark \\
	\hline
\end{array}
$$
where we write $\H=\H_A \otimes \H_B$ to denote that the Hilbert space can be decomposed according to subregions. Other notions will be introduced shortly. As we will see, the relative entropy has an algebraic generalization (denoted as $S_{\text {rel }}(\psi|\xi;\A)$) for factors of any type, while the generalized von Neumann entropy (denoted as $S(\psi;\A)$) can only be well-defined in type I/II due to the lack of well-defined notions of the trace function and density operators in factor of type III.

%
%

\subsection{Modular theory}

\begin{defn}
	A subset $\mathcal{H}_{0}\subset\mathcal{H}$ is dense in $\mathcal{H}$ if for every vector $\ket{\psi} \in \mathcal{H}$ and for every $\epsilon>0$, there exists a vector $\ket{\phi} \in \mathcal{H}_{0}$ such that $\|\ket{\psi}-\ket{\phi}\|<\epsilon$.
\end{defn}

\begin{defn}
	$\ket{\psi}\in\H$ is cyclic with respect to a von Neumann algebra $\A$ if $\A\ket{\psi}\equiv\{\O\ket{\psi}|\forall \O\in \A\}$ is dense in $\H$.
\end{defn}

The fact that $\mathcal{H}_{0}\subset\mathcal{H}$ is dense in $\mathcal{H}$ implies that every vector in $\H$ can be approximated by a vector in $\mathcal{H}_{0}$. More rigorously, for any vector $\ket{\psi}\in\H$, there exists a sequence $\{\ket{\phi_n}|n=1,2,\dots\}\subset\H_0$ such that $\lim_{n\rightarrow\infty}\ket{\phi_n}=\ket{\psi}$. Therefore, the fact that $\ket{\psi}\in\H$ is cyclic with respect to $\A$ says that $\H$ can be generated by $\A$ acting on $\ket{\psi}$, i.e. $\ket{\psi}$ cycles through $\H$ via $\A$.

\begin{defn}
	$\ket{\psi}\in\H$ is separating with respect to a von Neumann algebra $\A$ if $\O\ket{\psi}=0 $ implies $\O=0$ for $\O\in\A$.
\end{defn}

Provided a separating vector $\ket{\psi}\in\H$, if we have $A\ket{\psi}=B\ket{\psi}$ with $A,B\in\A$, then we have
\begin{equation}
	(A-B)\ket{\psi}=0\quad\Rightarrow\quad A-B=0\quad \text{or}\quad A=B
\end{equation}
i.e. two distinct operators in $M$ can not act in the same way on $\ket{\psi}$, so that any separating vector with respect to $\A$ can be used to distinguish or separate operators in $\A$.

%
%
%

\begin{thm}\label{thm:cs}
	$\ket{\psi}\in\H$ is separating with respect to $\A$ if and only if $\ket{\psi}\in\H$ is cyclic with respect to $\A'$, and vice versa.
\end{thm}

\begin{proof}
	\begin{itemize}
		\item ``if'': Consider $\O\in\A$ satisfying $\O\ket{\psi}=0$, then 
		\begin{equation}
			\O\A'\ket{\psi}=\A'\O\ket{\psi}=0
		\end{equation}
		Since $\A'\ket{\psi}$ is dense in $\H$, we have $\O$ annihilates all states in $\H$ such that $\O=0$.
		
		
		\item ``only if'': If $\ket{\psi}\in\H$ is not cyclic with respect to $\A'$, then $\A'\ket{\psi}$ is a proper subspace of $\H$. Consider a projection 
		\begin{equation}
			\Pi:\H\rightarrow\H \quad \text{onto}\quad (\A'\ket{\psi})_{\perp}\quad \Rightarrow \quad \Pi\A'\ket{\psi}=0,
		\end{equation}
		which implies
		\begin{equation}
			\mathbb{I}\in \A'\quad \Rightarrow \quad \ket{\psi}=\mathbb{I}\ket{\psi}\in \A'\ket{\psi}\quad \Rightarrow \quad \Pi\ket{\psi}=0.
		\end{equation}
		However, we can find $\Pi$ is bounded and commutes with $\A'$, so $\Pi\in\A$, i.e. there exists a nonzero operator in $\A$ which annihilates $\ket{\psi}$, so $\ket{\psi}$ is not separating w.r.t. $\A$.
	\end{itemize}
	
\end{proof}

A direct consequence is that when we assume $\ket{\psi}\in\H$ is both cyclic and separating with respect to $\A$, then $\ket{\psi}\in\H$ is also both cyclic and separating with respect to $\A'$.

\begin{defn}
	A relative Tomita operator on $\A$ is an anti-linear operator satisfying\footnote{Some literatures use the convention that interchanges the positions of $\psi$ and $\xi$ in the subscript of $S$.}
	\begin{equation}
		S_{\xi|\psi}\left(\O\ket{\psi}\right)=\O^\dagger \ket{\xi},\quad \forall \O\in\A
	\end{equation}
\end{defn}

Notice that $S_{\xi|\psi}$ is densely defined (i.e. whose domain is a dense subset of $\H$) if and only if $\ket{\psi}$ is cyclic and separating with respect to $\A$. The cyclic condition ensures that the domain $\O\ket{\psi}$ is dense while the separating condition is to avoid the possibility that $\O\ket{\psi}=0,~ \O^\dagger\ket{\xi}\neq 0$. Hereafter we mostly assume the cyclic separating condition of $\ket{\psi}$ for $S_{\xi|\psi}$.

\begin{thm}
	Provided both $\ket{\psi},\ket{\xi}$ are cyclic and separating with respect to $\A$, we have
	\begin{equation}
		S_{\xi|\psi}^{-1}=S_{\psi|\xi},
	\end{equation}
\end{thm}

\begin{proof}
	$S_{\psi|\xi}S_{\xi|\psi}\left(\O\ket{\psi}\right)=S_{\psi|\xi}\left(\O^\dagger \ket{\xi}\right)=\O\ket{\psi},~\forall \O\in\A$.
\end{proof}

\begin{thm}
	Provided $\ket{\psi}$ is cyclic and separating with respect to $\A$, we have
	\begin{equation}
		S^\dagger_{\xi|\psi}=S'_{\xi|\psi}
	\end{equation}
	where $S'_{\xi|\psi}$ is a relative Tomita operator on $\A'$.
\end{thm}

\begin{proof}
	Note that a proper definition of anti-linear operators acting on the bra first is
	\begin{equation}
		\left(\langle\phi|S\right)|\chi\rangle\equiv \left[\langle\phi|\left(S|\chi\rangle\right)\right]^*
	\end{equation}
	which implies
	\begin{equation}\label{eq:anti-linear}
		\diracprod{\phi}{S}{\chi}=\left[\left(\langle\phi|S\right)|\chi\rangle\right]^*=\left[\braket{S^\dagger\phi}{\chi}\right]^*=\braket{\chi}{S^\dagger\phi}=\diracprod{\chi}{S^\dagger}{\phi}
	\end{equation}
	then due to anti-linearity of the relative Tomita operator, we only need to prove that
	\begin{equation}
		\diracprod{\chi}{S'_{\xi|\psi}}{\phi}=\diracprod{\phi}{S_{\xi|\psi}}{\chi}
	\end{equation}
	
	Since by assumption and theorem \ref{thm:cs}, $\ket{\psi}$ is cyclic with respect to both $\A$ and $\A'$, we can set $\ket{\phi}=\O'\ket{\psi},~\ket{\chi}=\O\ket{\psi}$ where $\O\in\A,~\O'\in\A'$ such that
	\begin{align}
		&\diracprod{\chi}{S'_{\xi|\psi}}{\phi}=\diracprod{\psi}{\O^\dagger \O'^\dagger}{\xi}=\diracprod{\psi}{\O'^\dagger \O^\dagger}{\xi}=\diracprod{\phi}{S_{\xi|\psi}}{\chi}
	\end{align}
	where the second equality uses that $\O\in\A,~\O'\in\A'\Rightarrow [\O,\O']=0\Rightarrow[\O^\dagger,\O'^\dagger]=0$.
\end{proof}

\begin{defn}
	Provided $\ket{\psi}$ is cyclic and separating with respect to $\A$, the relative modular operator on $\A$ is defined by\footnote{There is an equivalent definition of the relative modular operator via the unique polar decomposition $S_{\xi|\psi}=J_{\xi|\psi}\Delta_{\xi|\psi}^{\frac{1}{2}}$ with $J_{\xi|\psi}$ anti-unitary.}
	\begin{equation}
		\Delta_{\xi|\psi}\equiv S^\dagger_{\xi|\psi}S_{\xi|\psi}
	\end{equation}
	and the relative modular Hamiltonian on $\A$ is defined by 
	\begin{equation}
		h_{\xi|\psi}\equiv-\log \Delta_{\xi|\psi}
	\end{equation}
\end{defn}

\begin{thm}
	The relative modular operator $\Delta_{\xi|\psi}$ on $\A$ is Hermitian.
\end{thm}

\begin{proof}
	$\Delta_{\xi|\psi}^\dagger=\left(S^\dagger_{\xi|\psi}S_{\xi|\psi}\right)^\dagger=S^\dagger_{\xi|\psi}S_{\xi|\psi}$.
\end{proof}

\begin{thm}
	Provided both $\ket{\psi},\ket{\xi}$ are cyclic and separating with respect to $\A$, we have
	\begin{equation}
		\Delta_{\psi|\xi}^{-1}=\Delta'_{\xi|\psi}\quad\Rightarrow\quad h_{\psi|\xi}=-h_{\xi|\psi}'
	\end{equation}
	where $\Delta'_{\xi|\psi}$ and $h_{\xi|\psi}'$ are the relative modular operator and the relative modular Hamiltonian on $\A'$ respectively.
\end{thm}

\begin{proof}
	$\Delta'_{\xi|\psi}=S'^\dagger_{\xi|\psi}S'_{\xi|\psi}=S_{\xi|\psi}S^\dagger_{\xi|\psi}=S^{-1}_{\psi|\xi}(S^\dagger_{\psi|\xi})^{-1}=(S^\dagger_{\psi|\xi}S_{\psi|\xi})^{-1}=\Delta_{\psi|\xi}^{-1}$.
\end{proof}

\begin{defn}
	Provided $\ket{\psi}$ is cyclic and separating with respect to $\A$, we define the Tomita operator $S_{\psi}\equiv S_{\psi|\psi}:\O\ket{\psi}\mapsto\O^\dagger \ket{\psi}$; the modular operator $\Delta_{\psi}\equiv\Delta_{\psi|\psi}= S^\dagger_{\psi}S_{\psi}$; the modular Hamiltonian $h_{\psi}\equiv h_{\psi|\psi}=-\log \Delta_{\psi}$.
\end{defn}

\begin{thm}
	\begin{equation}
		\diracprod{\psi}{\O\P}{\psi}=\diracprod{\psi}{\P\Delta_{\psi}\O}{\psi},~\forall\O,\P\in\A
	\end{equation}
\end{thm}

\begin{proof}
	\begin{align*}
		\diracprod{\psi}{\P\Delta_{\psi}\O}{\psi}&=\diracprod{\psi}{\P S^\dagger_{\psi}S_{\psi}\O}{\psi}=\diracprod{\P^\dagger\psi}{ S^\dagger_{\psi}}{\O^\dagger\psi}\\
		&=\diracprod{\O^\dagger\psi}{ S_{\psi}}{\P^\dagger\psi}=\braket{\O^\dagger\psi}{\P\psi}=\diracprod{\psi}{\O\P}{\psi}
	\end{align*}
	where in the second line we use \eqref{eq:anti-linear}.
\end{proof}

\section{Algebraic entropies}\label{sec:alg entropies}

In this section, we discuss the algebraic version of relative entropy and von Neumann entropy. Before introducing these algebraic entropies, we relate several seemingly different definitions of ``state'' which frequently occur in literatures from different fields.

\subsection{Definitions of ``state''}\label{sec:def of state}

There are three definitions of ``state'':
\begin{enumerate}
	\item A state is a vector $\ket{\psi}$ normalized as $\braket{\psi}{\psi}=1$ in a Hilbert space $\H$.

	\item A state is a density operator (or density matrix) $\rho_{\psi;\A}\in \A$ satisfying 
	\begin{equation}
		\rho_{\psi;\A}=\rho_{\psi;\A}^\dagger;\quad \rho_{\psi;\A}\geq 0;\quad\Tr_{\A}\rho_{\psi;\A}=1.
	\end{equation}

	\item A state is a linear functional $\omega_\psi:\A\rightarrow \mathbb{C}$ satisfying
	\begin{equation}
		\omega_\psi\geq0;\quad \omega_\psi(\mathbb{I})=1. 
	\end{equation}
\end{enumerate}

The first two definitions are commonly used in quantum physics, with the second one being more frequent in quantum information. The third definition is less common and mathematical, which readers may encounter in issues with algebraic quantum field theory. These three definitions are correlated to each other via the expectation value: 
\begin{equation}\label{eq:diff def of state}
	\diracprod{\psi}{\O}{\psi}=\Tr_{\A}\left(\rho_{\psi;\A}\O\right)=\omega_\psi(\O),\quad \forall\O\in\A.
\end{equation}
Hereafter we will denote a state $\psi$ and not distinguish the above definitions. Note that the second definition based on density operators is not well-defined in factors of type III as briefly mentioned in the last of subsection \ref{subsec:vN alg}.

\subsection{Araki's relative entropy}

The quantum relative entropy is defined by
\begin{equation}\label{def:quantum relative entropy}
	\Srel(\rho_\psi|\rho_\phi)\equiv\Tr\left[\rho_{\psi}\left(\log\rho_{\psi}-\log\rho_{\phi}\right)\right]
\end{equation}
which measures how much the state $\psi$ differs from another state $\phi$. Its algebraic generalization is defined due to Araki \cite{araki1975relative,araki1975inequalities}:
\begin{equation}\label{def:Arakis relative entropy}
	\Srel(\psi|\phi;\A)\equiv\diracprod{\psi}{h_{\phi|\psi}}{\psi}
\end{equation}
with $\ket{\psi}$ being cyclic and separating with respect to $\A$.

To see that Araki's relative entropy \eqref{def:Arakis relative entropy} is indeed a generalization of the quantum relative entropy \eqref{def:quantum relative entropy}, we now give a finite-dimensional example (factor of type I) \cite{Witten:2018zxz} where the two relative entropies coincide with each other. Consider a finite-dimensional bipartite system $\H=\H_A\otimes\H_B$ with $\dim A=\dim B=n$, we define a von Neumann algebra $\A$ to be of form
\begin{equation}
	\A=\{\O_A \otimes \mathbb{I}_B|\forall \O_A\in \B(\H_A)\}=\B(\H_A)\otimes \mathbb{I}_B
\end{equation}
then the commutant is given by
\begin{equation}
	\A'=\{\mathbb{I}_A \otimes \O_B|\forall \O_B\in \B(\H_B)\}=\mathbb{I}_A \otimes\B(\H_B)
\end{equation}

A lessen here is that identifying the algebras of localized observables $\A,\A'$ is equivalent to identifying a ``decomposition'' of the total system just like $\H=\H_A\otimes\H_B$, while the former is well-defined in factors of any type. In this sense, von Neumann algebras are associated with subregions of the system. Symbolically, we write 
\begin{equation}
	\A \sim A,\quad \A' \sim B,\quad \Tr_{\A}\sim \Tr_{A},\quad \Tr_{\A'}\sim\Tr_B.
\end{equation}

Next, we consider two vectors in $\H$:
\begin{equation}\label{eq:two vecs}
	\ket{\psi}=\sum_{k=1}^{n}c_k \ket{\psi_k}_A\otimes\ket{\psi_k'}_B,\quad \ket{\phi}=\sum_{\alpha=1}^{n}d_\alpha \ket{\phi_\alpha}_A\otimes\ket{\phi_\alpha'}_B
\end{equation}
with all $c_k$ being nonzero. The condition that $c_k\neq 0$ implies that if $\O_A \otimes \mathbb{I}_B\in \A$ annihilates $\ket{\psi}$, then $\O_A \otimes \mathbb{I}_B$ must annihilates all $\ket{\psi_k}_A$, which further implies that $\O_A \otimes \mathbb{I}_B=0$, i.e. $\ket{\psi}$ is separating with respect to $\A$. Likewise, we can argue that $\ket{\psi}$ is separating with respect to $\A'$. Recall that $\ket{\psi}$ is separating with respect to $\A'$ if and only if $\ket{\psi}$ is cyclic with respect to $\A$, we now have $\ket{\psi}$ is cyclic and separating with respect to $\A$. In this case, $\Srel(\psi|\phi;\A)$ is well-defined.

To compute Araki's relative entropy $\Srel(\psi|\phi;\A)$, we first compute the relative Tomita operator $S_{\phi|\psi}$ and then $\Delta_{\phi|\psi},h_{\phi|\psi}$. Consider an operator $\O_A$ defined by
\begin{equation}
	\O_A\ket{\psi_i}_A=\ket{\phi_\alpha}_A,\quad \O_A\ket{\psi_j}_A=0\text{ if }j\neq i
\end{equation}
whose adjoint is given by
\begin{equation}
	\O_A^\dagger\ket{\phi_\alpha}_A=\ket{\psi_i}_A,~ \O_A^\dagger\ket{\phi_\beta}_A=0\text{ if }\beta\neq \alpha
\end{equation}
then $\O_A$ acting on \eqref{eq:two vecs} results in
\begin{equation}
	\Rightarrow\quad (\O_A \otimes \mathbb{I}_B)\ket{\psi}=c_i\ket{\phi_\alpha}_A\otimes\ket{\psi'_i}_B,~(\O_A^\dagger \otimes \mathbb{I}_B)\ket{\phi}=d_\alpha\ket{\psi_i}_A\otimes\ket{\phi'_\alpha}_B
\end{equation}

According to the definition of $S_{\phi|\psi}$, we have
\begin{equation}
	\begin{aligned}
		&S_{\phi|\psi}\left[(\O_A \otimes \mathbb{I}_B)\ket{\psi}\right]=(\O_A^\dagger \otimes \mathbb{I}_B)\ket{\phi}\\
		\Rightarrow\quad &S_{\phi|\psi}\left(\ket{\phi_\alpha}_A\otimes\ket{\psi'_i}_B\right)=\frac{d_\alpha}{c_i^*}\ket{\psi_i}_A\otimes\ket{\phi'_\alpha}_B
	\end{aligned}
\end{equation}
where we can see the season why $c_i$ should be nonzero, otherwise the coefficients of $S_{\phi|\psi}$ may diverge. To compute $\Delta_{\phi|\psi}=S^\dagger_{\phi|\psi}S_{\phi|\psi}$, we make use of anti-linearity \eqref{eq:anti-linear} such that
\begin{equation}
	\begin{aligned}
		&\left(\bra{\psi_i}_A\otimes\bra{\phi'_\alpha}_B\right)S_{\phi|\psi}\left(\ket{\phi_\alpha}_A\otimes\ket{\psi'_i}_B\right)\\
		&=\frac{d_\alpha}{c_i^*}=\left(\bra{\phi_\alpha}_A\otimes\bra{\psi'_i}_B\right)S^\dagger_{\phi|\psi}\left(\ket{\psi_i}_A\otimes\ket{\phi'_\alpha}_B\right)
	\end{aligned}
\end{equation}

\begin{equation}
	\Rightarrow\quad S^\dagger_{\phi|\psi}\left(\ket{\psi_i}_A\otimes\ket{\phi'_\alpha}_B\right)=\frac{d_\alpha}{c_i^*}\ket{\phi_\alpha}_A\otimes\ket{\psi'_i}_B
\end{equation}

\begin{equation}
	\Rightarrow\quad\Delta_{\phi|\psi}\left(\ket{\phi_\alpha}_A\otimes\ket{\psi'_i}_B\right)=\frac{|d_\alpha|^2}{|c_i|^2}\ket{\phi_\alpha}_A\otimes\ket{\psi'_i}_B
\end{equation}

\begin{equation}
	\Rightarrow\quad\Delta_{\phi|\psi}=\left(\sum_{\alpha=1}^{n}|d_\alpha|^2\ket{\phi_\alpha}_A\bra{\phi_\alpha}_A\right)\otimes\left(\sum_{i=1}^{n}\frac{1}{|c_i|^2}\ket{\psi'_i}_B\bra{\psi'_i}_B\right)
\end{equation}

If we further define two pure states as follows,
\begin{equation}
	\rho_{\psi}\equiv \ket{\psi}\bra{\psi},\quad \rho_{\phi}\equiv \ket{\phi}\bra{\phi}
\end{equation}
whose reduced density operators are given by
\begin{equation}
	\begin{aligned}
		&\rho_{\psi;B}=\Tr_A\rho_{\psi}=\sum_{i=1}^{n}|c_i|^2\ket{\psi'_i}_B\bra{\psi'_i}_B,\\ &\rho_{\phi;A}=\Tr_B\rho_{\phi}=\sum_{\alpha=1}^{n}|d_\alpha|^2\ket{\phi_\alpha}_A\bra{\phi_\alpha}_A
	\end{aligned}
\end{equation}

\begin{equation}\label{eq:Hamiltonian splittable fin}
	\Rightarrow\quad \Delta_{\phi|\psi}=\rho_{\phi;A}\otimes\rho_{\psi;B}^{-1}
\end{equation}

Recall that previously we associate von Neumann algebras with subregions. In this case we have 
\begin{equation}
	\rho_{\phi;\A}\sim \rho_{\phi;A},\quad \rho_{\psi;\A'}\sim \rho_{\psi;B}
\end{equation}
which reformulates \eqref{eq:Hamiltonian splittable fin} into
\begin{equation}\label{eq:Hamiltonian splittable infin}
	\Delta_{\phi|\psi}=\rho_{\phi;\A}\otimes\rho_{\psi;\A'}^{-1}
\end{equation}
This relation holds generally in factors of type I/II \cite{Chandrasekaran:2022eqq}. Physically, \eqref{eq:Hamiltonian splittable infin} implies that the usual modular Hamiltonian of a physical system is splittable\footnote{The splittable condition is crucial in deriving the generalized second law for the generalized entropy in gravitational background \cite{Jensen:2023yxy}.}. Regarding our finite-dimensional setup, we have
\begin{equation}
	\begin{aligned}
		\Delta_{\phi|\psi}&=\rho_{\phi;A}\otimes\rho_{\psi;B}^{-1}\\
		\Rightarrow\quad h_{\phi|\psi}&=-\log \Delta_{\phi|\psi}=-\left(\log\rho_{\phi;A}\otimes\mathbb{I}_B-\mathbb{I}_A\otimes \log\rho_{\psi;B}\right)\\
		&=H_{\phi;A}-H_{\psi;B}
	\end{aligned}
\end{equation}
where $H_{\phi;A}\equiv -\log\rho_{\phi;A},~H_{\psi;B}\equiv -\log\rho_{\psi;B}$ are the usual modular Hamiltonians frequently appeared in quantum information. Now we continue our computation for Araki's relative entropy:
\begin{equation}
	\begin{aligned}
		\diracprod{\psi}{h_{\phi|\psi}}{\psi}&=-\Tr\left[\rho_{\psi}\left(\log\rho_{\phi;A}\otimes\mathbb{I}_B-\mathbb{I}_A\otimes \log\rho_{\psi;B}\right)\right]\\
		&=-\Tr_A\left(\rho_{\psi;A}\log\rho_{\phi;A}\right)+\Tr_B\left(\rho_{\psi;B}\log\rho_{\psi;B}\right)
	\end{aligned}
\end{equation}
where one notices that the second term is the von Neumann entropy $S(\rho_{\psi;B})$ which is equal to $S(\rho_{\psi;A})$ since $\rho_{\psi}$ is a pure state. We therefore have
\begin{equation}
	\begin{aligned}
		\diracprod{\psi}{h_{\phi|\psi}}{\psi}&=-\Tr_A\left(\rho_{\psi;A}\log\rho_{\phi;A}\right)+\Tr_A\left(\rho_{\psi;A}\log\rho_{\psi;A}\right)\\
		&=\Tr_A\left[\rho_{\psi;A}\left(\log\rho_{\psi;A}-\log\rho_{\phi;A}\right)\right]
	\end{aligned}
\end{equation}
which compares with \eqref{def:quantum relative entropy} such that we have
\begin{equation}
	\Srel(\psi|\phi;\A)\equiv\diracprod{\psi}{h_{\phi|\psi}}{\psi}=\Srel(\rho_{\psi;A}|\rho_{\phi;A})
\end{equation}
i.e. Araki's relative entropy reduces to quantum relative entropy between states restricted in subregion $A$, with which the algebra $\A$ of localized observables is associated.

\subsection{Algebraic von Neumann entropy}\label{sec:algvNalg}

The von Neumann entropy is defined by
\begin{equation}
	S(\rho_{\psi})\equiv-\Tr\left(\rho_{\psi}\log\rho_{\psi}\right)
\end{equation}
whose algebraic generalization in factors of type I/II is given by \cite{segal1960note,ohya2004quantum,Longo:2022lod}:
\begin{equation}\label{eq:algvnentropy}
	S(\psi;\A)\equiv-\Srel(\psi|\tau;\A)=-\diracprod{\psi}{h_{\tau|\psi}}{\psi}
\end{equation}
where $\ket{\psi}$ is cyclic and separating with respect to $\A$ and $\tau$ is a tracial state. A tracial state is a state $\ket{\tau}\in\H$ satisfying
\begin{equation}
	\diracprod{\tau}{\O\P}{\tau}=\diracprod{\tau}{\P\O}{\tau},\quad \forall\O,\P\in\A
\end{equation}
and each tracial state defines a trace function on $\A$:
\begin{equation}\label{def:trace}
	\Tr_{\A}^\tau(\O)\equiv\diracprod{\tau}{\O}{\tau}
\end{equation}

The non-existence of algebraic von Neumann entropy in factors of type III is due to the non-existence of tracial state or trace function. Hereafter we abuse the notation to use $\tau$ to denote the tracial states on both $\A$ and $\A'$ (which are in general not the same state) for simplicity, and one should distinguish them from the relevant context.

Recall that \eqref{eq:diff def of state} relates the vector representation of a state and its operator representation via $\diracprod{\tau}{\O}{\tau}=\Tr_{\A}^\tau(\rho_{\tau;\A}\O)$, we have $\rho_{\tau;\A}=\mathbb{I}$ which implies that a tracial state is a maximally mixed state. One may notice that the notation $\rho_{\tau;\A}=\mathbb{I}$ does not coincide with the usual finite-dimensional density matrix $\rho_{\text{max}}=\frac{1}{n}I_n$, which is due to that the trace we define on $\A$ is not same as the usual trace on the matrix. They only differ by a rescaling, and the only normalization condition for density operator is $\Tr_{\A}\rho_{\A}=1$. The advantage of $\rho_{\tau;\A}=\mathbb{I}$ is that it is well-defined in infinite-dimension, while $\frac{1}{n}I_n$ becomes ill-defined when $n\rightarrow\infty$. In fact, there is a theorem ensuring that every factor (not of type III) admits a unique faithful trace up to a rescaling \cite{Sorce:2023fdx,Longo:2022lod}\footnote{\label{fn:rescale}If we rescale the trace by a constant $C>0$: 
	\begin{equation}
		\overline{\Tr}_{\A}=\frac{1}{C}\Tr_{\A}\quad \Rightarrow\quad \bar{\rho}_{\psi;\A}=C \rho_{\psi;\A}
	\end{equation}
	which is to ensure that $\overline{\Tr}_{\A}\bar{\rho}_{\psi;\A}=1$. 	Such rescaling shifts the von Neumann entropy:
	\begin{equation}
		\begin{aligned}
			\bar{S}(\bar{\rho}_{\psi;\A})&\equiv -\overline{\Tr}_{\A}\left(\bar{\rho}_{\psi;\A}\log\bar{\rho}_{\psi;\A}\right)\\
			&=-\underbrace{\left(\Tr_{\A}\rho_{\psi;\A}\right)}_{=1}\log C -\Tr_{\A}\left(\rho_{\psi;\A}\log \rho_{\psi;\A}\right)\\
			&=S(\rho_{\psi;\A})-\log C
		\end{aligned}
	\end{equation}
	It seems like only the difference of von Neumann entropy is physically sensible, which we will come back to in the discussion.}. In this sense, we drop the superscript of the trace function defined in \eqref{def:trace}.

To compare the von Neumann entropy with the algebraic von Neumann entropy, we need to compute $h_{\tau|\psi}$. We start with
\begin{equation}
	\Delta_{\tau|\psi}=S^\dagger_{\tau|\psi}S_{\tau|\psi},\quad S_{\tau|\psi}\ket{\psi}=\ket{\tau},
\end{equation}
then we have
\begin{align}
	\diracprod{\psi}{\Delta_{\tau|\psi}}{\psi}=\diracprod{\psi}{S^\dagger_{\tau|\psi}S_{\tau|\psi}}{\psi}=\diracprod{\psi}{S^\dagger_{\tau|\psi}}{\tau}=\diracprod{\tau}{S_{\tau|\psi}}{\psi}=\braket{\tau}{\tau}
\end{align}
on the other hand, we have
\begin{align}
	\diracprod{\psi}{\Delta_{\tau|\psi}}{\psi}=\Tr_{\A}\left(\rho_{\psi;\A}\Delta_{\tau|\psi}\right)=\diracprod{\tau}{\rho_{\psi;\A}\Delta_{\tau|\psi}}{\tau}
\end{align}
where we use $\Tr_{\A}\left(\rho_{\psi;\A}\O\right)=\diracprod{\psi}{\O}{\psi},~\Tr_{\A}(\O)=\diracprod{\tau}{\O}{\tau}$ successively. We therefore have
\begin{equation}
	\braket{\tau}{\tau}=\diracprod{\tau}{\rho_{\psi;\A}\Delta_{\tau|\psi}}{\tau}
\end{equation}
which implies
\begin{equation}\label{eq:htau}
	\mathbb{I}=\rho_{\psi;\A}\Delta_{\tau|\psi}\quad \Rightarrow\quad \Delta_{\tau|\psi}=\rho_{\psi;\A}^{-1}\quad \Rightarrow\quad h_{\tau|\psi}=\log \rho_{\psi;\A}
\end{equation}
Finally,
\begin{equation}
	\begin{aligned}
		\Srel(\psi|\tau;\A)&=\diracprod{\psi}{h_{\tau|\psi}}{\psi}=\diracprod{\psi}{\log\rho_{\psi;\A}}{\psi}\\
		&=\Tr_{\A}\left(\rho_{\psi;\A}\log\rho_{\psi;\A}\right)\equiv-S(\rho_{\psi;\A})
	\end{aligned}
\end{equation}
\begin{equation}
	\Rightarrow\quad S(\psi;\A)\equiv-\Srel(\psi|\tau;\A)=S(\rho_{\psi;\A})
\end{equation}
i.e. the algebraic von Neumann entropy in $\A$ coincides with the von Neumann entropy of reduced density operator restricted in a subregion associated with the algebra $\A$ of localized observables.

Now we give a comment about the coincidence between the entropies and their algebraic generalization. Recall the splittable condition of modular Hamiltonian \eqref{eq:Hamiltonian splittable infin} holds for factors of type I/II, then we use the trace to rewrite
\begin{equation}
	\Delta_{\phi|\psi}=\rho_{\phi;\A}\otimes\rho^{-1}_{\psi;\A'}=\Delta^{-1}_{\tau|\phi}\otimes \Delta_{\tau|\psi}'~\Rightarrow~ h_{\phi|\psi}=-\log\Delta_{\phi|\psi}=h_{\tau|\psi}'-h_{\tau|\phi},
\end{equation} 
which implies
\begin{equation}
	\Srel(\psi|\phi;\A)=-S(\psi;\A')-\diracprod{\psi}{h_{\tau|\phi}}{\psi},
\end{equation}
which is in fact an algebraic generalization of a relation between quantum relative entropy and von Neumann entropy:
\begin{equation}
	\Srel(\rho_{\psi;A}|\rho_{\phi;A})=-S(\rho_{\psi;A})+\diracprod{\psi}{H_{\phi;A}}{\psi}=-S(\rho_{\psi;B})+\diracprod{\psi}{H_{\phi;A}}{\psi},
\end{equation}
where $H_{\phi;A}\equiv -\log\rho_{\phi;A} (=\log\Delta_{\phi;\A}=-h_{\tau|\phi})$.

\subsection{Extended definitions of ``state''}\label{sec:extended def}

In the beginning of the section, we introduce three definitions of ``state'' which are related to each other via the expectation value \eqref{eq:diff def of state}. Note however that, these three definitions are not completely equivalent. In the vector representation, one can always compute the expectation value of an operator, even in the case that the operator is not an element of the corresponding algebra, i.e. one can compute the expectation $\diracprod{\psi}{\O}{\psi},~\O\notin\A$ which makes no sense for the other two definitions. To be specific, the expectation value $\Tr_{\A}\left(\rho_{\psi;\A}\O\right),~\O\notin\A$ is ill-defined since $\rho_{\psi;\A}\O$ may not belongs to $\A$ then it is not in the domain of the trace function of $\A$. As for $\omega_\psi(\O),~\O\notin\A$, it is ill-defined since $\omega_\psi$ is only a linear functional on $\A$.

Nevertheless, we can extend the last two definitions via \eqref{eq:diff def of state} such that the expectation values of the last two definitions are well-defined for a larger class of operators. In the functional representation, we can simply define that
\begin{equation}
	\omega_\psi(\O)\equiv \diracprod{\psi}{\O}{\psi},~\O\notin\A.
\end{equation}
As for the density operator representation, we only need to extend the domain of the trace function via
\begin{equation}
	\Tr_{\A}\left(\rho_{\psi;\A}\O\right)\equiv\diracprod{\psi}{\O}{\psi},~\O\notin\A.
\end{equation}

Previously, we know that neither the (relative) Tomita operator nor the (relative) modular operator (Hamiltonian) are elements of the algebra. The former is due to the anti-linearity, and the latter is due to the unboundedness. Therefore, we could not regard the Araki's relative entropy and the algebraic von Neumann entropy as the expectation values of some operators in the sense of state functional. Now we resolve this issue via the above extension such that
\begin{equation}\label{def:Arakis relative entropy2}
	\Srel(\psi|\phi;\A)=\diracprod{\psi}{h_{\phi|\psi}}{\psi}=\omega_\psi(h_{\phi|\psi}),
\end{equation}
and
\begin{equation}\label{eq:algvnentropy2}
	S(\psi;\A)=-\diracprod{\psi}{h_{\tau|\psi}}{\psi}=-\omega_\psi(h_{\tau|\psi}).
\end{equation}

\subsection{Algebraic entropy difference}

Before our discussions about the algebraic reconstruction theorem, we give a comment about the algebraic version of the von Neumann entropy difference with respect to the variation of state:
\begin{equation}\label{eq:entropy diff}
	S(\rho+\delta\rho)-S(\rho)\approx-\Tr(\delta\rho\log\rho).
\end{equation}

Algebraically, we introduce a perturbed state $\psi$ by perturbating the state $\phi$ by a small variation along another state $\theta$ as follows,
\begin{equation}\label{eq:op variation}
	\omega_\psi\equiv\omega_\phi+\varepsilon(\omega_\theta-\omega_\phi)
\end{equation}
where $\varepsilon$ is an infinitesimal parameter. Note that \eqref{eq:op variation} is a functional valid for all operators after the extension in subsection \ref{sec:extended def}. The corresponding density operators are related as follows,
\begin{equation}\label{eq:op variation2}
	\omega_\psi(\O)=\omega_\phi(\O)+\varepsilon(\omega_\theta-\omega_\phi)(\O)
\end{equation}
\begin{equation}
	\Rightarrow\quad \Tr_{\A}(\rho_{\psi;\A}\O)=\Tr_{\A}(\rho_{\phi;\A}\O)+\varepsilon\Tr_{\A}(\rho_{\theta;\A}\O)-\varepsilon\Tr_{\A}(\rho_{\phi;\A}\O)
\end{equation}
for any arbitrary operator $\O$, which implies that
\begin{equation}\label{eq:density op diff}
	\rho_{\psi;\A}=\rho_{\phi;\A}+\varepsilon(\rho_{\theta;\A}-\rho_{\phi;\A}).
\end{equation}
\begin{equation}
	\Rightarrow\quad \log \rho_{\psi;\A}=\log \rho_{\phi;\A}+\varepsilon \rho_{\phi;\A}^{-1}(\rho_{\theta;\A}-\rho_{\phi;\A})+O(\varepsilon^2)
\end{equation}
where we use the formula of the operator logarithm $\log(A+\varepsilon B)=\log A+\varepsilon A^{-1}B+O(\varepsilon^2)$. Recall that we have \eqref{eq:htau}, we then have
\begin{equation}\label{eq:htau2}
	h_{\tau|\psi}=h_{\tau|\phi}+\varepsilon \rho_{\phi;\A}^{-1}(\rho_{\theta;\A}-\rho_{\phi;\A})+O(\varepsilon^2).
\end{equation}

We next define the algebraic entropy difference as
\begin{equation}
	\Delta S(\psi,\phi;\A)\equiv S(\psi;\A)-S(\phi;\A).
\end{equation}
Using \eqref{eq:algvnentropy2}, \eqref{eq:op variation} and \eqref{eq:htau2}, we find that
\begin{equation}
	\begin{aligned}
		S(\psi;\A)&=-\omega_\psi(h_{\tau|\psi})\\
		&\approx-\left[\omega_\phi+\varepsilon(\omega_\theta-\omega_\phi)\right]\left[ h_{\tau|\phi}+\varepsilon \rho_{\phi;\A}^{-1}(\rho_{\theta;\A}-\rho_{\phi;\A})\right]\\
		&\approx -\omega_\phi(h_{\tau|\phi})+\varepsilon(\omega_\theta-\omega_\phi)(h_{\tau|\phi})\\
		&~~~~+\varepsilon \omega_\phi\left[ \rho_{\phi;\A}^{-1}(\rho_{\theta;\A}-\rho_{\phi;\A})\right]
	\end{aligned}
\end{equation}
where we omit $O(\varepsilon^2)$ terms. Note the first term is exactly $S(\phi;\A)$ and the third term vanishes because
\begin{equation}
	\begin{aligned}
		\omega_\phi\left[ \rho_{\phi;\A}^{-1}(\rho_{\theta;\A}-\rho_{\phi;\A})\right]&=\Tr_{\A}\left[\rho_{\phi;\A}\cdot \rho_{\phi;\A}^{-1}(\rho_{\theta;\A}-\rho_{\phi;\A})\right]\\
		&=\Tr_{\A}(\rho_{\theta;\A}-\rho_{\phi;\A})\\
		&=\Tr_{\A}(\rho_{\theta;\A})-\Tr_{\A}(\rho_{\phi;\A})\\
		&=1-1=0
	\end{aligned}
\end{equation}
where the last line is due to that all density operators have unital trace. We therefore have
\begin{equation}\label{eq:entropy diff2}
	\Delta S(\psi,\phi;\A)=\varepsilon(\omega_\theta-\omega_\phi)(h_{\tau|\phi}).
\end{equation}
to first order. To see that it generalizes \eqref{eq:entropy diff}, we rewrite \eqref{eq:density op diff} as
\begin{equation}\label{eq:density op diff2}
	\rho_{\psi;\A}=\rho_{\phi;\A}+\delta\rho_{\phi;\A},\quad \delta\rho_{\phi;\A}\equiv \varepsilon(\rho_{\theta;\A}-\rho_{\phi;\A})
\end{equation}
such that
\begin{equation}
	\begin{aligned}
		\Delta S(\psi,\phi;\A)&=\varepsilon(\omega_\theta-\omega_\phi)(h_{\tau|\phi})\\
		&=\Tr_{\A}\left[ \varepsilon( \rho_{\theta;\A}-\rho_{\phi;\A}) h_{\tau|\phi}\right]\\
		&=\Tr_{\A}\left( \delta\rho_{\phi;\A}h_{\tau|\phi}\right)\\
		&=\Tr_{\A}\left( \delta\rho_{\phi;\A}\log \rho_{\phi;\A}\right).
	\end{aligned}
\end{equation}
where the last equality is due to \eqref{eq:htau}.

\section{Reconstruction theorems}\label{sec:rec}

\subsection{Finite-dimensional reconstruction theorem}

We first review the finite-dimensional reconstruction theorem, then we will present its algebraic generalization in type I/II factors in the next subsection.

\begin{thm}
	Given a finite-dimensional system $\H_{phys}=\H_{A}\otimes\H_{\bar{A}}$ and a code space $\H_{code}=\H_{a}\otimes\H_{\bar{a}}$ with conditions $\abs{a}\leq\abs{A},\abs{\bar{a}}\leq\abs{\bar{A}}$, the isometry\footnote{An isometry $V: \H_A\rightarrow \H_B~(\abs{A}\leq \abs{B})$ satisfies $V^\dagger V=I_A$ and $VV^\dagger=\Pi_{B}$ which is a projector onto $V(\H_A)\subset\H_B$.} $V: \H_{code}\rightarrow \H_{phys}$ induces two quantum channels:
	\begin{align*}
		\N: ~\tilde{\rho}_a\mapsto \rho_A\equiv\Tr_{\bar{A}}\left(V\tilde{\rho}_aV^\dagger\right),\quad \N': ~\tilde{\rho}_{\bar{a}}\mapsto \rho_{\bar{A}}\equiv\Tr_{A}\left(V\tilde{\rho}_{\bar{a}}V^\dagger\right),
	\end{align*}
	
	then the following statements are equivalent:
	
	\begin{enumerate}
		
		\item Given $\widetilde{\O}_a$ on $\H_a$: $\exists \O_A\in\mathcal{L}(\H_A)\Rightarrow \O_AV|\tilde{\phi}\rangle=V\widetilde{\O}_a|\tilde{\phi}\rangle$; likewise for $\widetilde{\O}_{\bar{a}}$.
		
		\item Given $\tilde{\rho},\tilde{\sigma}$ on $\H_{code}$: $\Srel(\rho_A|\sigma_A)=\Srel(\tilde{\rho}_a|\tilde{\sigma}_a),~\Srel(\rho_{\bar{A}}|\sigma_{\bar{A}})=\Srel(\tilde{\rho}_{\bar{a}}|\tilde{\sigma}_{\bar{a}})$.
		
		\item 
		
		Given $\widetilde{\O}_{a}$ on $\H_a$: $[\widetilde{\O}_{a},V^\dagger \P _{\bar{A}}V]=0,\forall \P _{\bar{A}}\in\mathcal{L}(\H_{\bar{A}})$; likewise for $\widetilde{\O}_{\bar{a}}$.
		
		\item Given $\tilde{\rho}$ on $\H_{code}$: $S(\rho_A)=\mathcal{L}_A+S(\tilde{\rho}_a),~S(\rho_{\bar{A}})=\mathcal{L}_{\bar{A}}+S(\tilde{\rho}_{\bar{a}})$.

	\end{enumerate}
	
\end{thm}

\subsection{Algebraic reconstruction theorem for type I/II factors}\label{sec:Algrec}

To algebraically generalize the finite-dimensional reconstruction theorem into factors of type I/II, we first introduce our algebraic setup which basically follows \cite{Kang:2018xqy,Kang:2019dfi}, but our notations follow \cite{Harlow:2016vwg}:
\begin{itemize}
	\item Let $\A_{code},\A_{phys}$ be von Neumann factors of type I/II on $\H_{code},\H_{phys}$ respectively, with $\A'_{code},\A'_{phys}$ respectively being the commutants. Let $V: \H_{code}\rightarrow \H_{phys}$ be an isometry which induces von Neumann factors $V\A_{code}V^\dagger,V\A'_{code}V^\dagger$ of type I/II on the image of $V$ denoted as $\im(V)=V(\H_{code})\subset \H_{phys}$.

	\item Notations: vectors in $\H_{code}$ and operators in $\A_{code}$ are labeled by the tilde sign, and operators in commutants are labeled by the prime sign. For example, $\ket{\widetilde{\Psi}}\in\H_{code}$, $\widetilde{\O}\in \A_{code},\widetilde{\O}'\in\A'_{code}$. (For simplicity, relative operators on $\H_{code}$ are not labeled by the tilde sign, but one can tell from the states they apply.)

	\item Suppose that the set of cyclic and separating vectors w.r.t. $\A_{code}$ is dense in $\H_{code}$ ($\Leftrightarrow$ the set of cyclic and separating vectors w.r.t. $\A'_{code}$ is dense in $\H_{code}$).

	\item Suppose that if $\ket{\widetilde{\Psi}}\in\H_{code}$ is cyclic and separating w.r.t. $\A_{code}$, then $V\ket{\widetilde{\Psi}}\in\H_{phys}$ is cyclic and separating w.r.t. $\A_{phys}$.

\end{itemize}

\begin{thm}\label{thm:alg}
	The following statements are equivalent: 
	\begin{enumerate}
		\item For any $\widetilde{\O}\in\A_{code},\widetilde{\O}'\in\A_{code}'$, there exist $\O\in \A_{phys},\O'\in \A_{phys}'$ such that
		\begin{equation}
			V\widetilde{\O}\ket{\widetilde{\Psi}}=\O V\ket{\widetilde{\Psi}}, V\widetilde{\O}'\ket{\widetilde{\Psi}}=\O' V\ket{\widetilde{\Psi}},\forall \ket{\widetilde{\Psi}}\in\H_{code}.
		\end{equation}
		
		\item For any $\ket{\widetilde{\Psi}},\ket{\widetilde{\Phi}}\in\H_{code}$ with $\ket{\widetilde{\Psi}},\ket{\widetilde{\Phi}}$ both cyclic and separating w.r.t. $\A_{code}$, 
		\begin{equation}\label{eq:JLMS}
			\begin{aligned}
				&\Srel(\widetilde{\Psi}|\widetilde{\Phi};\A_{code})=\Srel(\Psi|\Phi;\A_{phys}),\\&\Srel(\widetilde{\Psi}|\widetilde{\Phi};\A'_{code})=\Srel(\Psi|\Phi;\A'_{phys}),
			\end{aligned}
		\end{equation}
		where $\ket{\Psi}\equiv V\ket{\widetilde{\Psi}},\ket{\Phi}\equiv V\ket{\widetilde{\Phi}}$.
		
		\item For any $\P \in\A_{phys},\P '\in\A_{phys}'$,
		\begin{equation}
			V^\dagger \P V\in\A_{code},\quad V^\dagger \P 'V\in\A_{code}'.
		\end{equation}

		\item

		For any cyclic and separating $\ket{\widetilde{\Psi}}\in\H_{code}$  w.r.t. $\A_{code}$, we have
		\begin{equation}\label{eq:statement 4}
			\begin{aligned}
				&S(\widetilde{\Psi};\A_{code})=S(\Psi;\A_{phys}),\\ &S(\widetilde{\Psi};\A_{code}')=S(\Psi;\A_{phys}').
			\end{aligned}
		\end{equation}
	\end{enumerate}
\end{thm}

The equivalences between the first three statements are rigorously proved in \cite{Kang:2018xqy} for factors of general type. We identify the last statement as the algebraic version of the Ryu-Takayanagi formula, and we now prove that the equivalence between the first three statements and the statement 4. Practically, we first introduce a new statement (which we call the statement 1$'$):
\begin{itemize}
	\item \textit{1$'$. We define $\A_{phys}|_{\im(V)}$ as the subalgebra of $\A_{phys}$ which acts on the invariant subspace $\im(V)\subset \H_{phys}$, and likewise for $\A'_{phys}|_{\im(V)}$. We have
		\begin{equation}
			\begin{aligned}
				&V\A_{code}V^\dagger\subset \A_{phys}|_{\im(V)},\\ &V\A'_{code}V^\dagger\subset \A'_{phys}|_{\im(V)}.
			\end{aligned}
	\end{equation}}
\end{itemize}
The statement $1'$ equivalently interprets the bulk reconstruction as the statement 1. Since $1\Leftrightarrow2$, then we are free to use the statement 2 when we assume that the statement 1 (or the statement $1'$) holds. Finally we prove that $(1\Leftrightarrow1'+ 2)\Rightarrow4\Rightarrow2$.

Before our proof, we argue several facts about the induced von Neumann factors $V\A_{code}V^\dagger,V\A'_{code}V^\dagger$ on $\im(V)$. Firstly, we argue that the tracial state $\ket{\tilde{\tau}}$ on $\A_{code}$ under isometry $V$ is a tracial state on $V\A_{code}V^\dagger$. We define
\begin{equation}
	\ket{\tau}\equiv V\ket{\tilde{\tau}}\in \im(V)
\end{equation} 
then for any $\O,\P\in V\A_{code}V^\dagger$, there exists $\widetilde{\O},\widetilde{\P}\in \A_{code}$ such that
\begin{equation}
	\O=V\widetilde{\O}V^\dagger,\quad \P=V\widetilde{\P}V^\dagger
\end{equation}
according to the definition of $V\A_{code}V^\dagger$. It implies that
\begin{equation}
	\begin{aligned}
		\diracprod{\tau}{\O\P}{\tau}&=\diracprod{\tilde{\tau}}{V^\dagger\O\P V}{\tilde{\tau}}\\
		&=\diracprod{\tilde{\tau}}{V^\dagger V\widetilde{\O}V^\dagger V\widetilde{\P}V^\dagger V}{\tilde{\tau}}\\
		&=\diracprod{\tilde{\tau}}{\widetilde{\O} \widetilde{\P}}{\tilde{\tau}}
	\end{aligned}
\end{equation}
where the third equality uses $V^\dagger V=\mathbb{I}$ on $\H_{code}$, then
\begin{equation}
	\diracprod{\tau}{\O\P}{\tau}=\diracprod{\tilde{\tau}}{\widetilde{\O} \widetilde{\P}}{\tilde{\tau}}=\diracprod{\tilde{\tau}}{\widetilde{\P}\widetilde{\O} }{\tilde{\tau}}=\diracprod{\tau}{\P\O}{\tau}
\end{equation}
so that $\ket{\tau}$ is tracial on $V\A_{code}V^\dagger$.

Secondly, we argue that if $\ket{\widetilde{\Psi}}\in\H_{code}$ is cyclic and separating w.r.t. $\A_{code}$, then $\ket{\Psi}\equiv V\ket{\widetilde{\Psi}}\in\im(V)$ is cyclic and separating w.r.t. $V\A_{code}V^\dagger$. To achieve that, we need to argue that $\ket{\Psi}$ possesses (a) cyclic property: $V\A_{code}V^\dagger\ket{\Psi}$ is dense in $\im(V)$; (b) separating property: $\O\ket{\Psi}=0$ implies $\O=0$ for $\O\in V\A_{code}V^\dagger$. For cyclic property, we have 
\begin{equation}
	V\A_{code}V^\dagger\ket{\Psi}=V\A_{code}V^\dagger V\ket{\widetilde{\Psi}}=V\A_{code}\ket{\widetilde{\Psi}}.
\end{equation}
Since $\ket{\widetilde{\Psi}}\in\H_{code}$ is cyclic w.r.t. $\A_{code}$, i.e. $\A_{code}\ket{\widetilde{\Psi}}$ is dense in $\H_{code}$, then we have $V\A_{code}\ket{\widetilde{\Psi}}$ is dense in $V(\H_{code})=\im(V)$. For separating property, we consider any $\O\in V\A_{code}V^\dagger$, there exists $\widetilde{\O}\in \A_{code}$ such that $\O=V\widetilde{\O}V^\dagger$. Now we suppose that $\O\ket{\Psi}=0$, and we have
\begin{equation}
	\O\ket{\Psi}=V\widetilde{\O}V^\dagger V\ket{\widetilde{\Psi}}=V\widetilde{\O}\ket{\widetilde{\Psi}}=0,
\end{equation}
and $\ket{\widetilde{\Psi}}\in\H_{code}$ is separating w.r.t. $\A_{code}$, i.e. $\widetilde{\O}\ket{\widetilde{\Psi}}=0$ implies $\widetilde{\O}=0$ for $\widetilde{\O}\in \A_{code}$. We must conclude that $\widetilde{\O}=0$ because $V\neq 0$, which further implies that $\O=V\widetilde{\O}V^\dagger=0$. These confirm our arguments that $\ket{\Psi}\in\im(V)$ is cyclic and separating w.r.t. $V\A_{code}V^\dagger$. Similar arguments for $\ket{\Psi}\in\im(V)$ to be cyclic and separating w.r.t. $V\A'_{code}V^\dagger$ also hold. Now we are ready for the main proof.

\begin{proof}

	\begin{itemize}
		\item $1\Rightarrow 1'$:
		
		If statement 1 holds, we have that for any $\widetilde{\O}\in\A_{code}$, there exists $\O\in \A_{phys}$ such that
		$V\widetilde{\O}\ket{\widetilde{\Psi}}=\O V\ket{\widetilde{\Psi}}$ for all $ \ket{\widetilde{\Psi}}\in\H_{code}$. Since we also have $V\widetilde{\O}\ket{\widetilde{\Psi}}=V\widetilde{\O}V^\dagger V\ket{\widetilde{\Psi}}=V\widetilde{\O}V^\dagger \ket{\Psi}$ and $\O V\ket{\widetilde{\Psi}}=\O\ket{\Psi}$ for all $\ket{\Psi}=V\ket{\widetilde{\Psi}}\in\im(V)$, then
		\begin{equation}
			V\widetilde{\O}\ket{\widetilde{\Psi}}=\O V\ket{\widetilde{\Psi}}\quad \Leftrightarrow \quad V\widetilde{\O}V^\dagger \ket{\Psi}=\O\ket{\Psi}.
		\end{equation}
		The last equation tells us that, on the LHS we have an operator in $V\A_{code} V^\dagger$ acting on a state in $\im(V)=V(\H_{code})$, which should result in a state in $\im(V)$. On the other hand, we have an operator in $\A_{phys}$ acting on a state in $\im(V)$ on the RHS. Due the equality of the two sides, we conclude that the RHS also gives a state in $\im(V)$. Since we choose $\ket{\Psi}$ to be arbitrary, we find $\im(V)$ is an invariant subspace for the operator $\O$ on the RHS, i.e. such $\O$ is an element in $\A_{phys}|_{\im(V)}$ which is a subalgebra of $\A_{phys}$ acting on the invariant subspace $\im(V)\subset \H_{phys}$.
		
		Also notice that the above argument holds for any $\widetilde{\O}\in \A_{code}$ and there is a one-to-one correspondence between $\A_{code}$ and $V\A_{code} V^\dagger$ via $\widetilde{\O}\in\A_{code}\leftrightarrow \widehat{O}=V\widetilde{\O}V^\dagger\in V\A_{code} V^\dagger$.
		Therefore, the statement 1 can be rewritten as follows: for any $\widehat{\O}\in V\A_{code} V^\dagger$, there exists $\O\in \A_{phys}|_{\im(V)}$ such that 
		\begin{equation}\label{eq:OhateqO}
			\widehat{\O}\ket{\Psi}=\O\ket{\Psi},\quad \forall \ket{\Psi}\in \im(V).
		\end{equation}

		Since \eqref{eq:OhateqO} holds for all states in $\im(V)$, $\widehat{\O}$ and $\O$ as bounded operators acting on $\im(V)$ must be identical.\footnote{Two bounded linear operators $A,B$ on a Hilbert space $\H$ are equal if and only if $A\ket{\psi}=B\ket{\psi}$ for all $\ket{\psi}\in\H$.} To conclude, the statement 1 implies that for any $\widehat{\O}\in V\A_{code} V^\dagger$, there always exists $\O\in \A_{phys}|_{\im(V)}$ such that $\widehat{\O}=\O$, which further implies that
		\begin{equation}
			V\A_{code} V^\dagger \subset \A_{phys}|_{\im(V)}.
		\end{equation}
		Likewise, we can prove that $V\A'_{code} V^\dagger \subset \A'_{phys}|_{\im(V)}$. 
		
		\item $1'\Rightarrow 1$:
		
		Trivial.
		
		\item $(1'+ 2)\Rightarrow4$:
		
		To deduce the equality between algebraic von Neumann entropies, we start with the statement 2 which states that 
		\begin{equation}\label{eq:Sreleq1}
			\Srel(\widetilde{\Psi}|\widetilde{\Phi};\A_{code})=\Srel(\Psi|\Phi;\A_{phys}).
		\end{equation}
		for any $\ket{\widetilde{\Psi}},\ket{\widetilde{\Phi}}\in\H_{code}$ with $\ket{\widetilde{\Psi}},\ket{\widetilde{\Phi}}$ both being cyclic and separating w.r.t. $\A_{code}$.
		
		Notice that on the RHS, we have
		\begin{equation}\label{eq:Sreleq2}
			\Srel(\Psi|\Phi;\A_{phys})=\Srel(\Psi|\Phi;V\A_{code} V^\dagger)
		\end{equation}
		in the case that $V\A_{code} V^\dagger \subset \A_{phys}|_{\im(V)}\subset\A_{phys}$, which is because (a) $\ket{\Psi}=V\ket{\widetilde{\Psi}}$ is cyclic and separating w.r.t. $V\A_{code}V^\dagger$ such that $\Srel(\Psi|\Phi;V\A_{code} V^\dagger)$ is well-defined; (b) both sides have the same expression $\diracprod{V\widetilde{\Psi}}{h_{V\widetilde{\Phi}|V\widetilde{\Psi}}}{V\widetilde{\Psi}}$ with $h_{V\widetilde{\Phi}|V\widetilde{\Psi}}$ on $V\A_{code} V^\dagger$ is well-defined as $h_{V\widetilde{\Phi}|V\widetilde{\Psi}}$ on $\A_{phys}$ restricted in $V\A_{code} V^\dagger \subset \A_{phys}|_{\im(V)}\subset\A_{phys}$. Now combining \eqref{eq:Sreleq1} and \eqref{eq:Sreleq2}, we have
		\begin{equation}
			\Srel(\widetilde{\Psi}|\widetilde{\Phi};\A_{code})=\Srel(\Psi|\Phi;V\A_{code} V^\dagger).
		\end{equation}
		Setting $\ket{\widetilde{\Phi}}=\ket{\widetilde{\tau}}$ with $\ket{\widetilde{\tau}}$ being the tracial state w.r.t. $\A_{code}$, and according to our first argument before the main proof which states that $\ket{\tau}\equiv V\ket{\tilde{\tau}}$ is the tracial state w.r.t. $V\A_{code} V^\dagger$, we arrive at
		\begin{equation}\label{eq:Sreleq3}
			S(\widetilde{\Psi};\A_{code})=S(\Psi;V\A_{code} V^\dagger)
		\end{equation}
		from the definition of the algebraic von Neumann entropy \eqref{eq:algvnentropy}. The explicit expression of the RHS is
		\begin{equation}
			S(\Psi;V\A_{code} V^\dagger)=-\diracprod{\Psi}{h_{\tau|\Psi}}{\Psi}=-\diracprod{\Psi}{\log \rho_{\Psi;V\A_{code} V^\dagger}}{\Psi}
		\end{equation} 
		where the second equality is due to \eqref{eq:htau}. To proceed, we notice that in the case of $V\A_{code} V^\dagger \subset\A_{phys}$, the density operator $\rho_{\Psi;V\A_{code} V^\dagger}\in V\A_{code} V^\dagger$ is also a density operator in $\A_{phys}$, i.e.
		\begin{equation}
			\rho_{\Psi;V\A_{code} V^\dagger}=\rho_{\Psi;\A_{phys}}.
		\end{equation}

		Formally we can confirm it as follows. In the case that $V\A_{code} V^\dagger \subset\A_{phys}$, the traces are related by
		\begin{equation}
			\begin{aligned}
				&\Tr_{V\A_{code}V^\dagger}=\Tr_{\A_{phys}}\big|_{V\A_{code}V^\dagger}\\
				\text{or}\quad &\Tr_{V\A_{code}V^\dagger}(\O)=\Tr_{\A_{phys}}(\O),~\forall \O\in V\A_{code}V^\dagger
			\end{aligned}
		\end{equation}
		due to the uniqueness of the trace function in $V\A_{code} V^\dagger$. Then for any operator $\O\in V\A_{code} V^\dagger $, we have
		\begin{equation}
			\Tr_{V\A_{code}V^\dagger}(\rho_{\Psi;V\A_{code} V^\dagger}\O)=\Tr_{\A_{phys}}(\rho_{\Psi;V\A_{code} V^\dagger}\O).
		\end{equation}
		On the other hand, we also have
		\begin{equation}
			\Tr_{V\A_{code}V^\dagger}(\rho_{\Psi;V\A_{code} V^\dagger}\O)=\diracprod{\Psi}{\O}{\Psi}=\Tr_{\A_{phys}}(\rho_{\Psi;\A_{phys}}\O)
		\end{equation}
		where the two equalities are both due to \eqref{eq:diff def of state} while the second one regards $\O$ as an element in $\A_{phys}$. Now the last two equations tell us that $\forall\O\in V\A_{code} V^\dagger\subset\A_{phys}$:
		\begin{equation}
			\Tr_{\A_{phys}}(\rho_{\Psi;V\A_{code} V^\dagger}\O)=\Tr_{\A_{phys}}(\rho_{\Psi;\A_{phys}}\O).
		\end{equation}
		Physically speaking, the two density operators $\rho_{\Psi;V\A_{code} V^\dagger}$ and $\rho_{\Psi;\A_{phys}}$ are experimentally indistinguishable for observables in $V\A_{code} V^\dagger\subset\A_{phys}$ acting on $\im(V)$. What follows is that
		\begin{equation}
			\begin{aligned}
				S(\Psi;V\A_{code} V^\dagger)&=-\diracprod{\Psi}{\log \rho_{\Psi;V\A_{code} V^\dagger}}{\Psi}\\
				&=-\diracprod{\Psi}{\log \rho_{\Psi;\A_{phys}}}{\Psi}\\
				&=S(\Psi;\A_{phys})
			\end{aligned}
		\end{equation}
		Combined with \eqref{eq:Sreleq3} we arrive at
		\begin{equation}
			S(\widetilde{\Psi};\A_{code})=S(\Psi;\A_{phys})
		\end{equation}
		
		Likewise for $S(\widetilde{\Psi};\A'_{code})=S(\Psi;\A'_{phys})$.

		\item $4\Rightarrow2$: 
		
		Since we are dealing with factors of type I/II, the splittable condition of modular Hamiltonian \eqref{eq:Hamiltonian splittable infin} holds,
		\begin{equation}
			\begin{aligned}
				&\Delta_{\widetilde{\Phi}|\widetilde{\Psi}}=\rho_{\widetilde{\Phi}}\otimes\rho'^{-1}_{\widetilde{\Psi}}=\Delta^{-1}_{\widetilde{\tau}|\widetilde{\Phi}}\otimes \Delta_{\widetilde{\tau}|\widetilde{\Psi}}'\\
				\Rightarrow\quad &h_{\widetilde{\Phi}|\widetilde{\Psi}}=-\log\Delta_{\widetilde{\Phi}|\widetilde{\Psi}}=h_{\widetilde{\tau}|\widetilde{\Psi}}'-h_{\widetilde{\tau}|\widetilde{\Phi}},
			\end{aligned}
		\end{equation} 
		where $\ket{\widetilde{\tau}}$ is tracial w.r.t. $\A_{code}$, which implies
		\begin{equation}
			\Srel(\widetilde{\Psi}|\widetilde{\Phi};\A_{code})=-S(\widetilde{\Psi};\A'_{code})-\diracprod{\widetilde{\Psi}}{h_{\widetilde{\tau}|\widetilde{\Phi}}}{\widetilde{\Psi}}.
		\end{equation} 
		Similarly,
		\begin{equation}
			\Srel(\Psi|\Phi;\A_{phys})=-S(\Psi;\A'_{phys})-\diracprod{\Psi}{h_{\bar{\tau }|\Phi}}{\Psi},
		\end{equation}
		where $\ket{\bar{\tau }}$ is tracial w.r.t. $\A_{phys}$. Subtracting the above two equations implies
		\begin{equation}
			\begin{aligned}
				&\Srel(\widetilde{\Psi}|\widetilde{\Phi};\A_{code})-\Srel(\Psi|\Phi;\A_{phys})\\
				&=S(\Psi;\A'_{phys})-S(\widetilde{\Psi};\A'_{code})\\
				&+\diracprod{\Psi}{h_{\bar{\tau }|\Phi}}{\Psi}-\diracprod{\widetilde{\Psi}}{h_{\widetilde{\tau}|\widetilde{\Phi}}}{\widetilde{\Psi}}.
			\end{aligned}
		\end{equation} 
		We find that if
		\begin{equation}
			\begin{aligned}
				&\diracprod{\Psi}{h_{\bar{\tau }|\Phi}}{\Psi}-\diracprod{\widetilde{\Psi}}{h_{\widetilde{\tau}|\widetilde{\Phi}}}{\widetilde{\Psi}}\\
				&=\diracprod{\widetilde{\Psi} }{\left(V^\dagger h_{\bar{\tau }|\Phi}V-h_{\widetilde{\tau}|\widetilde{\Phi}}\right)}{\widetilde{\Psi}}=0,
			\end{aligned}
		\end{equation}
		then we have
		\begin{equation}
			\begin{aligned}
				&S(\Psi;\A'_{phys})=S(\widetilde{\Psi};\A'_{code})\\ \Rightarrow\quad &\Srel(\widetilde{\Psi}|\widetilde{\Phi};\A_{code})=\Srel(\Psi|\Phi;\A_{phys}).
			\end{aligned}
		\end{equation}
		In fact, we can prove a stronger condition that
		\begin{equation}
			V^\dagger h_{\bar{\tau }|\Phi}V=h_{\widetilde{\tau}|\widetilde{\Phi}}.
		\end{equation}
		when the statement 4 holds. First, we consider the variation of state:
		\begin{equation}
			\omega_{\widetilde{\Psi}}=\omega_{\widetilde{\Phi}}+\varepsilon(\omega_{\widetilde{\Theta}}-\omega_{\widetilde{\Phi}})
		\end{equation}
		following \eqref{eq:op variation}, then the algebraic entropy difference for $\A_{code}$ is given by
		\begin{equation}\label{eq:delta Scode}
			\begin{aligned}
				\Delta S(\widetilde{\Psi},\widetilde{\Phi};\A_{code})&\equiv S(\widetilde{\Psi};\A_{code})-S(\widetilde{\Phi};\A_{code})\\
				&=\varepsilon(\omega_{\widetilde{\Theta}}-\omega_{\widetilde{\Phi}})(h_{\widetilde{\tau}|\widetilde{\Phi}})\\
				&=\varepsilon~\omega_{\widetilde{\Theta}}(h_{\widetilde{\tau}|\widetilde{\Phi}})+\varepsilon S(\widetilde{\Phi};\A_{code}).
			\end{aligned}
		\end{equation}
		following \eqref{eq:entropy diff2} and \eqref{eq:algvnentropy2}. Second, we find that $\omega_{\Psi}\equiv\omega_{V\widetilde{\Psi}}$ satisfies
		\begin{equation}\label{eq:omegaV}
			\omega_{\Psi}(\O)=\diracprod{V\widetilde{\Psi}}{\O}{V\widetilde{\Psi}}=\diracprod{\widetilde{\Psi}}{V^\dagger\O V}{\widetilde{\Psi}}=\omega_{\widetilde{\Psi}}(V^\dagger\O V)
		\end{equation}
		for any operator $\O$ (i.e. not necessarily in $V\A_{code}V^\dagger$) acting on $\H_{phys}$ after the extension of the domain of $\omega_{\widetilde{\Psi}}$, as discussed in subsection \ref{sec:extended def}. Therefore, we find that
		\begin{equation}
			\begin{aligned}
				\omega_{\Psi}(\O)&=\omega_{\widetilde{\Psi}}(V^\dagger\O V)\\
				&=\left[\omega_{\widetilde{\Phi}}+\varepsilon(\omega_{\widetilde{\Theta}}-\omega_{\widetilde{\Phi}})\right](V^\dagger\O V)\\
				&=\omega_{\widetilde{\Phi}}(V^\dagger\O V)+\varepsilon(\omega_{\widetilde{\Theta}}(V^\dagger\O V)-\omega_{\widetilde{\Phi}}(V^\dagger\O V))\\
				&=\omega_{\Phi}(\O )+\varepsilon(\omega_{\Theta}(\O )-\omega_{\Phi}(\O ))\\
				&=\left[\omega_{\Phi}+\varepsilon(\omega_{\Theta}-\omega_{\Phi})\right](\O )
			\end{aligned}
		\end{equation}
		i.e. we have
		\begin{equation}
			\omega_{\Psi}=\omega_{\Phi}+\varepsilon(\omega_{\Theta}-\omega_{\Phi})
		\end{equation}
		such that
		the algebraic entropy difference for $\A_{phys}$ is given by
		\begin{equation}\label{eq:delta Sphys}
			\begin{aligned}
				\Delta S(\Psi,\Phi;\A_{phys})&\equiv S(\Psi;\A_{phys})-S(\Phi;\A_{phys})\\
				&=\varepsilon(\omega_{\Theta}-\omega_{\Phi})(h_{\bar{\tau}|\Phi})\\
				&=\varepsilon~\omega_{\Theta}(h_{\bar{\tau}|\Phi})+\varepsilon S(\Phi;\A_{phys}).
			\end{aligned}
		\end{equation}
		Now combining the statement 4 \eqref{eq:statement 4} with \eqref{eq:delta Scode}, \eqref{eq:delta Sphys} and \eqref{eq:omegaV}, we find that
		\begin{equation}
			\varepsilon~\omega_{\widetilde{\Theta}}(h_{\widetilde{\tau}|\widetilde{\Phi}})=\varepsilon~\omega_{\Theta}(h_{\bar{\tau}|\Phi})=\varepsilon~\omega_{\widetilde{\Theta}}(V^\dagger h_{\bar{\tau}|\Phi}V)
		\end{equation}
		Since $\varepsilon$ is arbitrary infinitesimal and $\omega_{\widetilde{\Theta}}$ is an arbitrary state functional, we conclude that
		\begin{equation}
			V^\dagger h_{\bar{\tau }|\Phi}V=h_{\widetilde{\tau}|\widetilde{\Phi}}.
		\end{equation}
		Likewise, we can argue in the same way that
		\begin{equation}
			\Srel(\widetilde{\Psi}|\widetilde{\Phi};\A'_{code})=\Srel(\Psi|\Phi;\A'_{phys})
		\end{equation}
		which completes the proof.

	\end{itemize}
	
\end{proof}

\section{Discussion}\label{sec:discussions}
As mentioned in footnote \ref{fn:rescale}, if the trace on $\A_{phys}$ differs the trace associated with $\ket{\tau}$ by a rescaling, an extra constant appears and shifts the algebraic von Neumann entropy on $\A_{phys}$. Naively, it seems like the extra constant coincides with the area term in the RT formula, 
\begin{equation}
	S_{bdy}=\mathcal{L}+S_{bulk}
\end{equation}
which however is not the case. Unlike the shifting constant can be manually chosen to be as small as we want, the area term diverges. Furthermore, it diverges in the same way as the boundary von Neumann entropy does. If we rewrite the RT formula as
\begin{equation}
	S_{bdy}-\mathcal{L}=S_{bulk}
\end{equation}
then the leading divergence on the LHS cancels. Since recent investigations have argued for a algebraic way to regulate $S_{bulk}$ \cite{Witten:2021unn,Chandrasekaran:2022cip,Jensen:2023yxy,Kudler-Flam:2023hkl,Colafranceschi:2023urj,Faulkner:2024gst}, which implies that $S_{bdy}-\mathcal{L}$ should be regarded as the outcome after algebraic regulation, then we should have
\begin{equation}
	S_{bulk}=S(\tilde{\Psi};\A_{code}),\quad S_{bdy}-\mathcal{L}=S(\Psi;\A_{phys}).
\end{equation}
One future direction is to investigate how the area term $\mathcal{L}$ emerges in the algebraic regulation of $S_{bdy}$.

Another future direction that is worth pursuing is that whether a modified version of the algebraic RT formula is plausible in a general type factor. In \cite{Kudler-Flam:2023hkl}, the authors have argued that it is possible to rigorously define the difference of von Neumann entropy in a general factor, in the sense that the divergence of the algebraic von Neumann entropy cancels after subtraction in type III factors. In this case, it seems that we may promote the statement of algebraic RT formula in the algebraic reconstruction to a new statement regarding the generalized version of the difference of algebraic von Neumann entropy.

We also emphasize that for the AdS/CFT, the theorems in section \ref{sec:rec} only holds in $N\rightarrow\infty$ limit, in which case the quantum error correction is exact \cite{Cotler:2017erl}. For large but finite $N$, the quantum error correction in the theorems is approximate, e.g. \eqref{eq:JLMS} receives $O(1/N)$ corrections \cite{Cotler:2017erl,Jafferis:2015del} and the relations in the other statements should only approximately hold. How such corrections modify the algebraic RT formula is also worth further investigation.

\acknowledgments

We thank Yixu Wang, Antony J. Speranza, Qiang Wen, Yiwei Zhong for helpful discussions. The authors thank the Shing-Tung Yau Center of Southeast University for support. HZ is supported by SEU Innovation Capability Enhancement Plan for Doctoral Students (Grant No.CXJH\_SEU 24137).

\bibliographystyle{JHEP}
\bibliography{bib}
\end{document}